\documentclass[a4paper]{article}

\usepackage[margin=1in]{geometry}
\usepackage[pdftex,hyperindex,breaklinks]{hyperref}
\usepackage{amsmath,amsthm,bm}
\usepackage{algorithm}
\usepackage[noend]{algpseudocode}
\algnewcommand\Input{\item[\textbf{Input:}]}
\algnewcommand\Output{\item[\textbf{Output:}]}
\usepackage{xspace}
\usepackage[english]{babel}
\usepackage[utf8]{inputenc}
\usepackage[bitstream-charter]{mathdesign}

\newtheorem{thm}{Theorem}[section]
\newtheorem{lem}[thm]{Lemma}
\newtheorem{cor}[thm]{Corollary}
\newtheorem{prop}[thm]{Proposition}
\newtheorem{rem}[thm]{Remark}
\theoremstyle{remark}
\newtheorem{example}{Example}

\newcommand{\gO}{O}
\newcommand{\gOt}{\tilde{O}}
\newcommand{\gOeps}{\gO_\epsilon}
\newcommand{\gOteps}{\gOt_\epsilon}
\newcommand{\cfq}{{\mathbb{F}_{\!q}}}
\newcommand{\cfqs}{{\mathbb{F}_{\!q^s}}}
\newcommand{\N}{\ensuremath{\mathbb{N}}\xspace}
\newcommand{\gz}{\mathbb{Z}}
\newcommand{\gr}{R}
\newcommand{\verif}{\textsc{VerifySP}\xspace}
\newcommand{\verifsum}{\textsc{VerifySumSP}\xspace}
\newcommand{\spr}{\textsc{SparseProduct}\xspace}
\newcommand{\ip}{\textsc{InterpSumSP}\xspace}
\newcommand{\ft}{\textsc{FindTerms}\xspace}
\newcommand{\rp}{\textsc{RandomPrime}\xspace}
\newcommand{\sparsityEstimate}{\textsc{SparsityEstimate}\xspace}
\newcommand{\mul}{\mathsf{M}}
\newcommand{\imul}{\mathsf{I}}

\DeclareMathOperator{\supp}{supp}
\DeclareMathOperator\proba{Pr}

\begin{document}
\title{Essentially Optimal Sparse Polynomial Multiplication}
\author{Pascal Giorgi \hfill Bruno Grenet\hfill Armelle Perret du Cray\\
    LIRMM, Univ. Montpellier, CNRS\\
    Montpellier, France\\
    \{pascal.giorgi,bruno.grenet,armelle.perret-du-cray\}@lirmm.fr}

\maketitle

\begin{abstract}
  We present a probabilistic algorithm to compute the product of two univariate sparse polynomials over a field with a number of bit
  operations that is quasi-linear in the size of the input and the output. Our algorithm works for any field of characteristic zero or
  larger than the degree. We mainly rely on sparse interpolation and on a new algorithm for verifying a sparse product that has also a
  quasi-linear time complexity. Using Kronecker substitution techniques we extend our result to the multivariate case.
\end{abstract}

\section{Introduction}

Polynomials are one of the most basic objects in computer algebra and the study of fast polynomial operations remains a very challenging
task. Polynomials can be represented using either the dense representation, that stores all the coefficients in a vector, or the more
compact sparse representation, that only stores nonzero monomials. In the dense representation, we know quasi-optimal algorithms for
decades. Yet, this is not the case for sparse polynomials.

In the sparse representation, a polynomial $F = \sum_{i=0}^D f_i X^i \in\gr[X]$ is expressed as a list of pairs $(e_i, f_{e_i})$ such that
all the $f_{e_i}$ are nonzero. We denote by $\#F$ its \emph{sparsity}, i.e. the number of nonzero coefficients. Let $F$ be a polynomial of
degree $D$, and $B$ a bound on the size of its coefficients. Then, the size of the sparse representation of $F$ is $\gO(\#F(B+\log D))$
bits. It is common to use $B = 1+\max_i(\lfloor\log_2(|f_{e_i}|)\rfloor)$ if $\gr = \gz$ and $B = 1+\lfloor\log_2 q\rfloor$ if $\gr = \cfq$.
The sparse representation naturally extends to polynomials in $n$ variables: Each exponent is replaced by a vector of exponents which gives
a total size of $\gO(\#F(B+n\log D)$.

Several problems on sparse polynomials have been investigated to design fast algorithms, including arithmetic operations, interpolation and
factorization. We refer the interested readers to the excellent survey by Roche and the references therein~\cite{Roche2018}. Contrary to the
dense case, note that \emph{fast} algorithms for sparse polynomials have a \mbox{(poly-)}logarithmic dependency on the degree.
Unfortunately, as shown by several $\mathsf{NP}$-hardness results, such fast algorithms might not even exist unless $\mathsf P=\mathsf{NP}$.
This is for instance the case for \textsc{gcd} computations~\cite{Plaisted84}.

In this paper, we are interested in the problem of sparse polynomial multiplication. In particular, we provide the first quasi-optimal
algorithm whose complexity is quasi-linear in both the input and the output sizes.

\subsection{Previous work}

The main difficulty and the most interesting aspect of sparse polynomial multiplication is the fact that the size of the output does not
exclusively depend on the size of the inputs, contrary to the dense case. Indeed, the product of two polynomials $F$ and $G$ has at most
$\#F\#G$ nonzero coefficients. But it may have as few as $2$ nonzero coefficients.

\begin{example}
  Let $F=X^{14}+2X^7 +2$, $G=3X^{13}+5X^8+3$ and $H=X^{14}- 2X^7 +2$. Then $FG=
  3X^{27}+5X^{22}+6X^{20}+10X^{15}+3X^{14}+6X^{13}+10X^8+6X^7+6$ has nine terms, while $FH=X^{28}+4$ has only two.
\end{example}

The product of two polynomials of sparsity $T$ can be computed by generating the $T^2$ possible monomials, sorting them by increasing degree
and merging those with the same degree. Using radix sort, this algorithm takes $\gO(T^2 (\mul_\gr +\log D))$ bit operations, where $\mul_R$
denotes the cost of one operation in $\gr$. A major drawback of this approach is its space complexity that exhibits a $T^2$ factor, even if
the result has less than $T^2$ terms. Many improvements have been proposed to reduce this space complexity, to extend the approach to
multivariate polynomials, and to provide fast implementations in practice~\cite{Johnson74, MonaganPearce2009, MonaganPearce2011}. Yet, none
of these results reduces the $T^2$ factor in the time complexity.

In general, no complexity improvement is expected as the output polynomial may have as many as $T^2$ nonzero coefficients. However, this
number of nonzero coefficients can be overestimated, giving the opportunity for output-sensitive algorithms. Such algorithms have first
been proposed for special cases. Notably, when the output size is known to be small due to sufficiently structured inputs~\cite{Roche2011},
especially in the multivariate case~\cite{vdHLec2012, vdHLebSch2013}, or when the support of the output is known in
advance~\cite{vdHLec2013}. It is possible to go one step further by studying the conditions for small outputs. A first reason is exponent
\emph{collisions}. Let $F = \sum_{i=1}^T f_i X^{\alpha_i}$ and $G = \sum_{j=1}^T g_j X^{\beta_j}$. A collision occurs when there exist
distinct pairs of indices $(i_1,j_1)$ and $(i_2,j_2)$ such that $\alpha_{i_1}+\beta_{j_1} = \alpha_{i_2}+\beta_{j_2}$. Such collisions
decrease the number of terms of the result. The second reason is coefficient cancellations. In the previous example, the resulting
coefficient is $(f_{i_1}g_{j_1} + f_{i_2}g_{j_2})$, which could vanish depending on the coefficient values. Taking into account the
exponent collisions amounts to computing the \emph{sumset} of the exponents of $F$ and $G$, that is $\{\alpha_i+\beta_j:1\le i,j\le T\}$.
Arnold and Roche call this set the \emph{structural support} of the product $FG$ and its size the \emph{structural
sparsity}~\cite{roche2015}. If $H = FG$, then the structural sparsity $S$ of the product $FG$ satisfies $2\le \#H \le S\le T^2$. Observe
that although $\#H$ and $S$ can be close, their difference can reach $\gO(T^2)$ as shown by the next example.

\begin{example}\label{ex:structuralsupport}
  Let $F = \sum_{i=0}^{T-1} X^i$, $G = \sum_{i=0}^{T-1} (X^{Ti+1}-X^{Ti})$ and $H = FG$. We have $\#F = T$, $\#G = 2T$ and the structural
  sparsity of $FG$ is $T^2+1$ while $H = X^{T^2}-1$ has sparsity $2$.
\end{example}

For polynomials with nonnegative integer coefficients, the support of $H$ is exactly the sumset of the exponents of $F$ and $G$, the
structural support of $H=FG$. In this case, Cole and Hariharan describe a multiplication algorithm requiring $\gOt(S\log^2
D)$\footnote{Here, and throughout the article, $\gOt(f(n))$ denotes $\gO(f(n)\log^k(f(n)))$ for some constant $k>0$.} operations in the RAM
model with $\gO(\log(CD))$ word size~\cite{ColeHariharan}, where $\log(C)$ bounds the bitsize of the coefficients. Arnold and Roche improve
this complexity to $\gOt(S\log D + \#H \log C)$ bit operations for polynomials with both positive and negative integer coefficients
~\cite{roche2015}. Note that they also extend their result to finite fields and to the multivariate case. A recent algorithm of Nakos
avoids the dependency on the structural sparsity for the case of integer polynomials~\cite{nakos2019}, using the same word RAM model as Cole
and Hariharan. Unfortunately, the bit complexity of this algorithm ($\gOt((T\log D+\#H\log^2 D) \log (CD)+\log^3 D)$) is not quasi-linear.

In the dense case, quasi-optimal multiplication algorithms rely on the well-known evaluation-interpolation scheme. In the sparse settings,
this approach is not efficient. The fastest multiplication algorithms mentioned above~\cite{roche2015,nakos2019} mainly rely on a different
method called \emph{sparse interpolation}\footnote{Despite their similar names, dense and sparse polynomial interpolation are actually two
quite different problems.}, that has received considerable attention. See e.g. the early results of Prony~\cite{Prony} and Ben-Or and
Tiwari~\cite{BenOrTiwari} or the recent results by Huang~\cite{huang2019}. Despite extensive analysis of this problem, no quasi-optimal
algorithm exists yet. We remark that it is not the only difficulty. Simply using a quasi-optimal sparse interpolation algorithm would not be
enough to get a quasi-optimal sparse multiplication algorithm~\cite{arnold2016}.

\subsection{Our contributions}

Our main result is summarized in Theorem~\ref{thm:product}. We extend the complexity notations to $\gOeps$ and $\gOteps$ for hiding some
polynomial factors in $\log(\frac{1}{\epsilon})$. Let $F = \sum_{i=1}^T f_i X^{e_i}$. We use $\|F\|_\infty = \max_i|f_i|$ to denote its
height, $\#F$ for its number of nonzero terms and $\supp(F)=\{e_1,\dots, e_T\}$ its support.

\begin{thm}\label{thm:product} \sloppy 
  Given two sparse polynomials $F$ and $G$ over $\gz$, Algorithm \spr computes $H = FG$ in $\gOt_\epsilon(T(\log D+\log C))$ bit operations
  with probability at least $1-\epsilon$, where $D = \deg(H)$, $C=\max(\|F\|_\infty, \|G\|_\infty, \|H\|_\infty)$ and $T = \max(\#F, \#G,
  \#H)$. The algorithm extends naturally to finite fields with characteristic larger than $D$ with the same complexity where $C$ denotes the
  cardinality.
\end{thm}

This result is based on two main ingredients. We adapt Huang's algorithm~\cite{huang2019} to interpolate $FG$ in quasi-linear time. Note
that the original algorithm does not reach quasi-linear complexity.

Sparse interpolation algorithms, including Huang's, require a bound on the sparsity of the result. We replaced this bound by a guess on the
sparsity and an \emph{a posteriori} verification of the product, as in~\cite{nakos2019}. However, using the classical polynomial evaluation
approach for the verification does not yield a quasi-linear bit complexity (see Section~\ref{sec:verif}). Therefore, we introduce a novel
verification method that is essentially optimal.

\begin{thm}\label{thm:verif} Given three sparse polynomials $F$, $G$ and $H$ over $\cfq$ or $\gz$, Algorithm \verif tests whether $FG = H$ in $\gOt_\epsilon(T(\log D + B))$ bit operations, where $D = \deg(H)$, $B$ is a bound on the bitsize of the coefficients of $F$, $G$ and $H$, and $T=\max(\#F,\#G,\#H)$. The answer is always correct if $FG = H$, and the probability of error is at most $\epsilon$ otherwise.
\end{thm}

Finally, using Kronecker substitution, we show that our sparse polynomial multiplication algorithm extends to the multivariate case with a
quasi-linear bit complexity $\gOteps(T(n\log d+B))$ where $n$ is the number of variables and $d$ the maximal partial degree on each
variable. Nevertheless, over finite fields this approach requires an exponentially large characteristic. Using the randomized Kronecker
substitution \cite{ArRo14} we derive a fast algorithm for finite fields of characteristic polynomial in the input size. Its bit complexity
is $\gOteps(nT(\log d+B))$. Even though it is not quasi-optimal, it achieves the best known complexity for this case.

\section{Preliminaries}\label{sec:prelim}

We denote by $\imul(n) = \gO(n\log n)$ the bit complexity of the multiplication of two integers of at most $n$
bits~\cite{vdHoevenHarvey2019}. Similarly, we denote by $\mul_q(D) = \gO(D\log(q)\log(D\log q)4^{\log^* D})$ the bit complexity of the
multiplication of two \emph{dense} polynomials of degree at most $D$ over $\cfq$ where $q$ is prime~\cite{vdHH2019:cyclomult}. The cost of
multiplying two elements of $\cfqs$ is $\gO(\mul_q(s))$. The cost of multiplying two dense polynomials over $\gz$ of heights at most $C$ and
degrees at most $D$ is $\mul_\gz(D, C) = \imul(D(\log C+\log D))$~\cite[Chapter~8]{MCAlgebra}.

\medskip
Since our algorithms use reductions \emph{modulo} $X^p-1$ for some prime number $p$, we first review useful related results.

\begin{thm}[Rosser and Schoenfeld~\cite{rosser1962}] \label{thm:rosser}
  If $\lambda\ge21$, there are at least $\frac{3}{5}\lambda/\ln\lambda$ prime numbers in $[\lambda,2\lambda]$.
\end{thm}

\begin{prop}[{\cite[Chapter 10]{shoup2008}}]\label{prop:rdprime}
  There exists an algorithm $\rp(\lambda,\epsilon)$ that returns an integer $p$ in $[\lambda, 2\lambda]$, such that $p$ is prime with
  probability at least $1-\epsilon$. Its bit complexity is $\gOt_\epsilon(\log^3\lambda)$.
\end{prop}

We need two distinct properties on the reductions \emph{modulo} $X^p-1$. The first one is classical in sparse interpolation to bound the
probability of exponent collision in the residue (see \cite[Lemma~3.3]{roche2015}).

\begin{prop}\label{prop:sanscollision}
  Let $H$ be a polynomial of degree at most $D$ and sparsity at most $T$, $0 < \epsilon< 1$ and $\lambda =\max(21,\frac{10}{3\epsilon}
  T^2\ln D)$. Then with probability at least $1-\epsilon$, $\rp(\lambda,\frac{\epsilon}{2})$ returns a prime number $p$ such that $H\bmod
  X^p-1$ has the same number of terms as $H$, that is no collision of exponents occurs.
\end{prop}

The second property allows to bound the probability that a polynomial vanishes \emph{modulo} $X^p-1$.

\begin{prop}\label{prop:choixdep}
  Let $H$ be a nonzero polynomial of degree at most $D$ and sparsity at most $T$, $0<\epsilon<1$ and $\lambda = \max(21,
  \frac{10}{3\epsilon} T\ln D)$. Then with probability at least $1-\epsilon$, $\rp(\lambda, \frac{\epsilon}{2})$ returns a prime number $p$
  such that $H\bmod X^p-1\neq 0$.
\end{prop}

\begin{proof}
  For $H\bmod X^p-1$ to be nonzero, it is sufficient that there exists one exponent $e$ of $H$ that is not congruent to any other exponent
  $e_j$ modulo $p$. In other words, it is sufficient that $p$ does not divide any of the $T-1$ differences $\delta_j=e_j-e$. Noting that
  $\delta_j\leq D$, the number of primes in $[\lambda,2\lambda]$ that divide at least one $\delta_j$ is at most $\frac{(T-1)\ln D}{\ln
  \lambda}$. Since there exist $\frac{3}{5}\lambda/\ln \lambda$ primes in this interval, the probability that a prime randomly chosen from
  it divides at least one $\delta_j$ is at most $\epsilon/2$. $\rp(\lambda,\epsilon/2)$ returns a prime in $[\lambda,2\lambda]$ with
  probability at least $1-\epsilon/2$, whence the result.
\end{proof}

The next two propositions are used to reduce integer coefficients modulo some prime number and to construct an extension field.

\begin{prop}\label{prop:choixdeq}
  Let $H\in\gz[X]$ be a nonzero polynomial, $0<\epsilon<1$ and $\lambda \geq \max(21,\frac{10}{3\epsilon}\ln \|H\|_\infty)$. Then with
  probability at least $1-\epsilon$, $\rp(\lambda, \frac{\epsilon}{2})$ returns a prime $q$ such that $H\bmod q \neq 0$.
\end{prop}

\begin{proof}
  Let $h_i$ be a nonzero coefficient of $H$. A random prime from $[\lambda,2\lambda]$ divides $h_i$ with probability at most $\frac{5}{3}
  \ln \|H\|_\infty/\lambda\leq \epsilon/2$. Since $\rp(\lambda,\epsilon/2)$ returns a prime in $[\lambda,2\lambda]$ with probability at
  least $1-\epsilon/2$ the result follows.
\end{proof}

\begin{prop}[{\cite[Chapter 20]{shoup2008}}]\label{prop:irrpoly}
  There exists an algorithm that, given a finite field $\cfq$, an integer $s$ and $0<\epsilon<1$, computes a degree-$s$ polynomial in
  $\cfq[X]$ that is irreducible with probability at least $1-\epsilon$. Its bit complexity is $\gOt_\epsilon(s^3\log q)$.
\end{prop}

\section{Sparse polynomial product verification}\label{sec:verif}

Verifying a product $FG = H$ of dense polynomials over an integral domain $\gr$ simply falls down to testing $F(\alpha)G(\alpha) =
H(\alpha)$ for some random point $\alpha\in\gr$. This approach exhibits an optimal linear number of operations in $\gr$ but it is not
deterministic. (No optimal deterministic algorithm exists yet.) When $\gr=\gz$ or $\cfq$, a divide and conquer approach provides a
quasi-linear complexity, namely $\gOt(DB)$ bit operations where $B$ bounds the bitsize of the coefficients.

For sparse polynomials with $T$ nonzero coefficients, evaluation is not quasi-linear since the input size is only $\gO(T(\log D+B))$.
Indeed, computing $\alpha^D$ requires $\gO(\log D)$ operations in $\gr$ which implies a bit complexity of $\gOt(\log(D)\log(q))$ when
$\gr=\cfq$. Applying this computation to the $T$ nonzero monomials gives a bit complexity of $\gOt(T\log(D)\log(q))$. We mention that the
latter approach can be improved to $\gOt((1+T/\log\log(D))\log(D)\log(q))$ using Yao's result \cite{Yao1976} on simultaneous exponentiation.
When $\gr=\gz$, the best known approach to avoid expression swell is to pick a random prime $p$ and to perform the evaluations modulo $p$.
One needs to choose $p>D$ in order to have a nonzero probability of success. Therefore, the bit complexity contains a $T\log^2 D$ factor.

Our approach to obtain a quasi-linear complexity is to perform the evaluation \emph{modulo} $X^p-1$ for some random prime $p$. This
requires to evaluate the polynomial $[(FG)\bmod X^p-1]$ on $\alpha$ without computing it.

\subsection{Modular product evaluation}\label{ssec:modbinôme}

\begin{lem} \label{lem:evalbin}
  Let $F$ and $G$ be two sparse polynomials in $\gr[X]$ with $\deg F,\deg G\leq p-1$ and $\alpha\in\gr$. Then $(FG)\bmod X^p-1$ can be
  evaluated on $\alpha$ using $\gO((\#F+\#G)\log p)$ operations in $\gr$.
\end{lem}

\begin{proof}
Let $H = (FG)\bmod X^p-1$. The computation of $H$ corresponds to the linear map
\[
  \underbrace{\begin{pmatrix} h_0\\h_1\\\vdots\\h_{p-1}\end{pmatrix}}_{\vec{h}}
= \underbrace{\begin{pmatrix}
    f_0 & f_{p-1} & \cdots & f_1\\
    f_1 & f_0     & \cdots & f_2\\
    \vdots & \vdots&&\vdots\\
    f_{p-1} & f_{p-2} & \cdots & f_0
    \end{pmatrix}}_{T_F} 
  \underbrace{\begin{pmatrix} g_0\\g_1\\\vdots\\g_{p-1}\end{pmatrix}}_{\vec{g}}
\]
where $f_i$ (resp. $g_i$, $h_i$) is the coefficient of degree $i$ of $F$ (resp. $G$, $H$). Computing $H(\alpha)$ corresponds to the inner
product $\vec\alpha_p \vec h=\vec\alpha_p T_F \vec g$ where $\vec\alpha_p = (1,\alpha,\dotsc,\alpha^{p-1})$. This evaluation can be
computed in $\gO(p)$ operations in $\gr$ \cite{giorgi2018}. Here we reuse similar techniques in the context of sparse polynomials.

To compute $H(\alpha)$, we first compute $\vec c = \vec\alpha_pT_F$, and then the inner product $\vec c\vec g$. If $\supp(G) = \{j_1,\dotsc,
j_{\#G}\}$ with $j_1 < \dotsb < j_{\#G} < p$, we only need the corresponding entries of $\vec c$, that is all $c_{j_k}$'s for $1\le
k\le\#G$. Since $c_j = \sum_{\ell=0}^{p-1} \alpha^\ell f_{(\ell-j)\bmod p}$, we can write $c_j = f_{p-j} + \alpha\sum_{\ell=0}^{p-2}
\alpha^\ell f_{(\ell-j+1)\bmod p}$, that is $c_j = \alpha c_{j-1} + (1 - \alpha^p) f_{p-j}$.

Applying this relation as many times as necessary, we obtain a relation to compute $c_{j_{k+1}}$ from $c_{j_k}$:
\[
  c_{j_{k+1}}=\alpha^{j_{k+1}-j_k}c_{j_k} + (1-\alpha^p)\sum_{\ell=j_k+1}^{j_{k+1}}\alpha^\ell f_{p-\ell}.
\]
Each nonzero coefficient $f_t$ of $F$ appears in the definition of $c_{j_{k+1}}$ if and only if $p-j_{k+1}\leq t< p-j_{k}$. Thus,
each $f_t$ is used exactly once to compute all the $c_{j_k}$'s. Since for each summand, one needs to compute $\alpha^\ell$ for
some $\ell<p$, the total cost for computing all the sums is $\gO(\#F\log p)$ operations in $\gr$. Similarly, the computation of
$\alpha^{j_{k+1}-j_k}c_{j_k}$ for all $k$ costs $\gO(\#G\log p)$. The last remaining step is the final inner product which costs
$\gO(\#G)$ operations in $\gr$, whence the result.
\end{proof}

The complexity is improved to $\gO(\log p + (\#F+\#G)\log p/\log\log p)$ using again Yao's algorithm \cite{Yao1976} for simultaneous
exponentiation.

\subsection{A quasi-linear time algorithm}\label{sec:algo}

Given three sparse polynomials $F$, $G$ and $H$ in $\gr[X]$, we want to assert that $H=FG$. Our approach is to take a random prime $p$ and
to verify this assertion modulo $X^p-1$ through modular product evaluation. This method is explicitly described in the algorithm \verif
that works over any large enough integral domain $\gr$. We further extend the description and the analysis of this algorithm for the
specific cases $\gr=\gz$ and $\gr=\cfq$ in the next sections.

\begin{algorithm}
\caption{\verif}
\begin{algorithmic}[1]
  \Input $H,F,G\in\gr[X]$; $0<\epsilon<1$.
  \Output True if $FG=H$, False with probability $\geq 1-\epsilon$ otherwise.
  \State Define $c_1>\frac{10}{3}$ and $c_2>1$ such that $\frac{10}{3c_1} + (1-\frac{10}{3c_1})\frac{1}{c_2}\leq\epsilon$
  \State $D\gets\deg(H)$
  \If{$\#H>\#F\#G$ or $D\neq \deg(F)+\deg(G)$} \Return False\EndIf
  \State $\lambda\gets\max(21,c_1(\#F\#G+\#H)\ln D)$
  \State $p\gets \rp(\lambda,\frac{5}{3c_1})$
  \State $(F_p,G_p,H_p) \gets(F \bmod X^p-1,~G \bmod X^p-1,~H\bmod X^p-1) $
  \State Define $\mathcal{E}\subset\gr$ of size $>c_2p$ and choose $\alpha\in\mathcal{E}$ randomly.
  \State $\beta \gets [(F_pG_p)\bmod X^p-1](\alpha)$ \Comment using Lemma \ref{lem:evalbin}
  \State \Return $\beta= H_p(\alpha)$
\end{algorithmic}
\end{algorithm}

\begin{thm}\label{thm:algogénéral} \sloppy
  If $\gr$ is an integral domain of size $\geq 2c_1c_2\#F\#G\ln D$ \verif works as specified and it requires
  $\gO_\epsilon(T\log(T\log D))$ operations in $\gr$ plus $\gO_\epsilon(T\imul(\log D))$ bit operations where $D=\deg(H)$ and
  $T=\max(\#F,\#G,\#H)$.
\end{thm}

\begin{proof}
  Step~3 dismisses two trivial mistakes and ensures that $D$ is a bound on the degree of each polynomial. If $FG=H$, the algorithm returns
  True for any choice of $p$ and $\alpha$. Otherwise, there are two sources of failure. Either $X^p-1$ divides $FG-H$, whence
  $(FG)_p(\alpha) = H_p(\alpha)$ for any $\alpha$. Or $\alpha$ is a root of the nonzero polynomial $(FG-H)\bmod X^p-1$. Since $FG-H$ has at
  most $\#F\#G+\#H$ terms, the first failure occurs with probability at most $\frac{10}{3c_1}$ by Proposition~\ref{prop:choixdep}. And since
  $(FG-H)\bmod X^p-1$ has degree at most $p-1$ and $\mathcal E$ has $c_2p$ points, the second failure occurs with probability at most
  $\frac{1}{c_2}$. Altogether, the failure probability is at most $\frac{10}{3c_1} + (1-\frac{10}{3c_1})\frac{1}{c_2}$.

  Let us remark that $c_1,c_2 = \gO(\frac{1}{\epsilon})$ and $p=\gO(\frac{1}{\epsilon}T^2\log D)$. Step~5 requires only
  $\gOt(\log^3(\frac{1}{\epsilon}T\log D))$ bit operations by Proposition~\ref{prop:rdprime}. The operations in Step~6 are $T$ divisions by
  $p$ on integers bounded by $D$ which cost $\gO_\epsilon(T\imul(\log D))$ bit operations, plus $T$ additions in $\gr$. The evaluation of
  $F_pG_p\bmod X^p-1$ on $\alpha$ at Step~8 requires $\gO(T\log(\frac{1}{\epsilon}T\log D))$ operations in $\gr$ by Lemma~\ref{lem:evalbin}.
  The evaluation of $H_p$ on $\alpha$ costs $\gO(T\log(T\log D))$ operations in $\gr$. Other steps have negligible costs.
\end{proof}

\subsection{Analysis over finite fields}

The first easy case is the case of large finite fields: If there are enough points for the evaluation, the generic algorithm has
the same guarantee of success and a quasi-linear time complexity.

\begin{cor}\label{cor:grandcf}
  Let $F$, $G$ and $H$ be three polynomials of degree at most $D$ and sparsity at most $T$ in $\cfq[X]$ where $q>2c_1c_2\#F\#G\ln(D)$. Then
  Algorithm \verif has bit complexity $\gO_\epsilon(n\log^2(n) 4^{\log^*n})$ where $n = T(\log D+\log q)$ is the input size.
\end{cor}

\begin{proof} \sloppy
  By definition of $n$, the cost of Step 6 is $\gO_\epsilon(n\log n)$ bit operations. Each ring operation in $\cfq$ costs
  $\gO(\log(q)\log\log(q) 4^{\log^*q})$ bit operations which implies that the bit complexity of Step 8 is $\gO_\epsilon(T\log(T\log
  D)\log(q)\log\log(q)4^{\log^*q})$. Since $T\log q$ and $T\log D$ are bounded by $n$ and $\log\log q \le \log n$, the result follows.
\end{proof}

We shall note that even if $q < 2c_1c_2\#F\#G\ln(D)$ we can make our algorithm to work by using an extension field and this approach
achieves the same complexity.

\begin{thm}
  One can adapt algorithm \verif to work over finite fields $\cfq$ such that $q < 2c_1c_2\#F\#G\ln(D)$. The bit complexity is
  $\gO_\epsilon(n\log(n)\log\log(n)\, 4^{\log^*n})$, where $n = T(\log D+\log q)$ is the input size.
\end{thm}

\begin{proof}
  To have enough elements in the set $\mathcal{E}$, we need to work over $\cfqs$ where $q^s > c_2p \ge q^{s-1}$. An irreducible degree-$s$
  polynomial can be computed in $\gOt(s^3\log q)=\gOt(\log (T \log D)/\log q)$ by Proposition~\ref{prop:irrpoly}. Since $\alpha$ is taken in
  $\cfqs$, the complexity becomes $\gOeps(T\imul(\log D)+T\log(T\log D)\mul_q(s))$ bit operations. Remarking that $T\le D$ we have $T\log
  (T\log D)\le T\log (D\log D) = \gO(n)$. Since $s\log q=\gO(\log(T\log D)) = \gO(\log n)$ we can obtain $\mul_q(s)=\gO(\log (n) \log \log
  (n)4^{\log^*n})$ which implies that the second term of the complexity is $\gO(n\log(n)\log\log(n)4^{\log^*n})$. The first term is
  negligible since it is $\gO(n\log n)$.

  In order to achieve the same probability of success, we fix an error probability $1/c_3<1$ for Proposition~\ref{prop:irrpoly} and we take
  constants $c_1$ and $c_2$ in \verif such that $1-(1-\frac{10}{3c_1})(1-\frac{1}{c_2})(1-\frac{1}{c_3})\le \epsilon$.
\end{proof}

We note that for very sparse polynomials over some fields, the complexity is only dominated by the operations on the exponents.

\begin{cor}
  \verif has bit complexity $\gO_\epsilon(n\log n)$ in the following cases :
  \begin{itemize}
  \item[(i)] $s=1$ and $\log q = \gO(\log^{1-\alpha} D)$ for some constant $0<\alpha<1$,
  \item[(ii)] $s>1$ and $T = \Theta(\log^k D)$ for some constant $k$.
  \end{itemize}
\end{cor}

\begin{proof}
  In both cases the cost of reducing the exponents \emph{modulo} $p$ is $\gOeps(n \log n)$ bit operations. In the first case, each
  multiplication in $\cfq$ costs $\gO(\log (q) \log \log (q) 4^{\log^* q})=\gO(\log D)$ bit operations as $\log\log(q) 4^{\log^*
  q}=\gO(\log^\alpha D)$. In the second case, $n = \gO(\log^{k+1}D)$ and $s\log q=\gOeps(\log(T^2 \log D))=\gOeps(\log\log D)$ which
  implies $\mul_q(s)=\gOeps(s\log(q)\log(s\log q)4^{\log^*s})=\gOeps(\log D)$. In both cases, the algorithm performs $\gOeps(T\log(T\log
  D))=\gOeps(T\log n)$ operations in $\cfq$ (or in $\cfqs$). Therefore the bit complexity is $\gOeps(n\log n)$.
\end{proof}

The following generalization is used in our quasi-linear multiplication algorithm given in Section~\ref{sec:product}.

\begin{cor}\label{cor:verifsumprodfq}
  Let $(F_i,G_i)_{0\leq i<m}$ and $H$ be sparse polynomials over $\cfq$ of degree at most $D$ and sparsity at most $T$. We can verify if
  $\sum_{i=0}^{m-1}F_iG_i=H$, with error probability at most $\epsilon$ when they are different, in $\gO_\epsilon(m(T\imul(\log
  D)+T\log(mT\log D)\mul_q(s)))$ bit operations.
\end{cor}

\subsection{Analysis over the integers}\label{ssec:z}

In order to keep a quasi-linear time complexity over the integers, we must work over a prime finite field $\cfq$ to avoid the computation of
too large integers. Indeed, $H_p(\alpha)$ could have size $p\log(\alpha)=\gO_\epsilon(T^2\log(D)\log(\alpha))$ which is not quasi-linear in
the input size.

\begin{thm} \label{thm:algosurz} \sloppy
  One can adapt algorithm \verif to work over the integers. The bit complexity is $\gOeps(n\log n\log\log n)$, where $n = T(\log D+\log C)$
  is the input size with $C = \max(\|F\|_\infty,\|G\|_\infty,\|H\|_\infty)$.
\end{thm}

\begin{proof} \sloppy
  Before Step 6, we choose a random prime number $q=\rp(\mu,\frac{5}{3c_2})$ with $\mu=c_2\max(p,\ln(C^2T+C))$ and we perform all the
  remaining steps modulo $q$. Let us assume that the polynomial $\Delta=FG-H \in \gz[X]$ is nonzero. Our algorithm only fails in the
  following three cases: $p$ is such that $\Delta_p=\Delta \bmod X^p-1=0$; $q$ is such that $\Delta_p \equiv 0 \bmod q$; $\alpha$ is a root
  of $\Delta_p$ in $\cfq$.

  Using Proposition~\ref{prop:choixdep}, $\Delta_p$ is nonzero with probability at least $1-\frac{10}{3c_1}$. Actually, with the same
  probability, the proof of the proposition shows that at least one coefficient of $\Delta$ is preserved in $\Delta_p$. Since
  $\|\Delta\|_\infty \leq C^2T+C$, Proposition~\ref{prop:choixdeq} ensures that $\Delta_p\not\equiv 0 \bmod q$ with probability at least
  $1-\frac{10}{3c_2}$. Finally, $q$ has been chosen so that $\cfq$ has at least $c_2p$ elements whence $\alpha$ is not a root of $\Delta_p
  \bmod q$ with probability at least $1-\frac{1}{c_2}$. Altogether, taking $c_1,c_2\ge \frac{10}{3}$ such that
  $1-(1-\frac{10}{3c_1})(1-\frac{10}{3c_2})(1-\frac{1}{c_2})\le\epsilon$, our adaptation of \verif has an error probability at most
  $\epsilon$.

  The reductions of $F$, $G$ and $H$ \emph{modulo} $q$ add a term $\gO(T\imul(\log C))$ to the complexity. Since operations in $\cfq$ have
  cost $\imul(\log q)$, the complexity becomes $\gO(T\imul(\log D)+T\imul(\log C) + T\log(T \log D)\imul(\log q))$ bit operations. The
  first two terms are in $\gO(n\log n)$. Moreover, $q = \gO_\epsilon(\log(C^2T)+p)$ and $p = \gO_\epsilon(T^2\log D)$, thus $\log
  q=\gO_\epsilon(\log(\log C+T\log D)) = \gO_\epsilon(\log n)$. Since $T\le D$, $T\log(T\log D) = \gO(n)$ and the third term in the
  complexity is $\gO_\epsilon(n\log n\log\log n)$.
\end{proof}

As over small finite fields, the complexity is actually better for very sparse polynomials.

\begin{cor}\label{cor:verifrèscreusedansz}
  If $T = \Theta(\log^k D)$ for some $k$, \verif has bit complexity $\gO_\epsilon(n\log n)$.
\end{cor}

\begin{proof}
  If $T=\Theta(\log^k D)$, $T\log(T\log D) = \gOt(\log^kD) = o(n)$, thus the last term of the complexity in the proof of
  Theorem~\ref{thm:algosurz} becomes negligible with respect to the first two terms.
\end{proof}

For the same reason as for finite fields, we extend the verification algorithm to a sum of products.

\begin{cor}\label{cor:verifsumprodz}
  Let $(F_i,G_i)_{0\leq i<m}$ and $H$ be sparse polynomials of degree at most $D$, sparsity at most $T$, and height at most $C$. We can
  verify if $\sum_{i=0}^{m-1}F_iG_i=H$, with probability of error at most $\epsilon$ when they are different, in $\gO_\epsilon(mT\imul(\log
  D)+mT\imul(\log C)+mT\log(mT\log D)\imul(\log(m\log C+mT\log D)))$ bit operations.
\end{cor}

We shall only use this algorithm with $m=2$ and thus refer to it as $\verifsum(H,F_0,G_0,F_1,G_1, \epsilon)$.

\section{Sparse polynomial multiplication}\label{sec:product}

Given two sparse polynomials $F$ and $G$, our algorithm aims at computing the product $H=FG$ through sparse polynomial interpolation. We
avoid the difficulty of computing an \emph{a priori} bound on the sparsity of $H$ needed for sparse interpolation by using our verification
algorithm of Section~\ref{sec:verif}. Indeed, one can start with an arbitrary small sparsity and double it until the interpolated polynomial
matches the product according to \verif.

The remaining difficulty is to interpolate $H$ in quasi-optimal time given a sparsity bound, which is not yet achieved in the general case.
In our case, we first analyze the complexity of Huang's sparse interpolation algorithm~\cite{huang2019} when the input is a sum of sparse
products. In order to obtain the desired complexity we develop a novel approach that interleaves two levels of Huang's algorithm.

\subsection{Analysis of Huang's sparse interpolation}\label{ssec:interpolation}

In \cite{huang2019} Huang proposes an algorithm that interpolates a sparse polynomial $H$ from its SLP representation, achieving the best
known complexity for this problem, though it is not optimal. Its main idea is to use the dense polynomials $H_p = H\bmod X^p-1$ and $H'_p =
H'\bmod X^p-1$ where $H'$ is the derivative of $H$ and $p$ a \emph{small} random prime. Indeed, if $cX^e$ is a term of $H$ that does not
collide during the reduction \emph{modulo} $X^p-1$, $H_p$ contains the monomial $cX^{e\bmod p}$ and $H'_p$ contains $ceX^{e-1\bmod p}$,
hence $c$ and $e$ can be recovered by a mere division. Of course, the choice of $p$ is crucial for the method to work. It must be small
enough to get a low complexity, but large enough for collisions to be sufficiently rare.

\begin{lem}\label{lem:findterms}
  There exists an algorithm \ft that takes as inputs a prime $p$, two polynomials $H_p = H\bmod X^p-1$, $H'_p = H'\bmod X^p-1$, and bounds
  $D\ge \deg(H)$ and $C\ge\|H\|_\infty$ and it outputs an approximation $H^*$ of $H$ that contains at least all the monomials of $H$ that do
  not collide \emph{modulo} $X^p-1$. Its bit complexity is $\gO(T\imul(\log CD))$, where $T = \#H$.
\end{lem}

\begin{proof}
  It is a straightforward adaptation of \cite[Algorithm 3.4 (\textsc{UTerms})]{huang2019}. Here, taking $C$ as input allows us to only
  recover coefficients that are at most $C$ in absolute value and therefore to perform divisions with integers of bitsize at most $
  \log(CD)$.
\end{proof}

\begin{cor}\label{cor:findterms} \sloppy
  Let $H$ be a sparse polynomial such that $\#H\le T$, $\deg H\le D$ and $\|H\|_\infty\le C$, and $0<\epsilon<1$. If $\lambda = \max(21,
  \frac{10}{3\epsilon} T^2\ln D)$ and $p=\rp(\lambda,\frac{\epsilon}{2})$, then with probability at least $1-\epsilon$, $\ft\,(p, H \bmod
  X^p-1, H'\bmod X^p-1, D, C)$ returns $H$.
\end{cor}

\begin{proof}
  With probability at least $1-\epsilon$, no collision occurs in $H \bmod X^p-1$, and consequently neither in $H' \bmod X^p-1$, by
  Proposition~\ref{prop:sanscollision}. In this case \ft correctly computes $H$, according to Lemma~\ref{lem:findterms}.
\end{proof}

\begin{thm}\label{thm:interpolation}
  There exists an algorithm \ip that takes as inputs $2m$ sparse polynomials $(F_i, G_i)_{0\le i<m}$, three bounds $T\ge \#H$, $D>\deg(H)$
  and $C\ge\|H\|_\infty$ where $H = \sum_{i=0}^{m-1} F_iG_i$, a constant $0 < \mu < 1$ and the list $\mathcal P$ of the first $2N$ primes
  for $N = \max(1,\lfloor\frac{32}{5}(T-1)\log D\rfloor)$, and outputs $H$ with probability at least $1-\mu$.

  Its bit complexity is $\gOt_\mu(mT_1\log(D_1)\log(C_1D_1))$ where $T_1$, $D_1$ and $C_1$ are bounds on the sparsity, the degree and the
  height of $H$ and each $F_i$ and $G_i$.
\end{thm}

\begin{proof}
  It is identical to the proof of \cite[Algorithm 3.9 (\textsc{UIPoly})]{huang2019} taking into account that $H$ is not given as an SLP
  anymore but as $\sum_{i=0}^{m-1}F_iG_i$ where the polynomials $F_i$ and $G_i$ are given as sparse polynomials.
\end{proof}

\begin{rem}\label{rem:complexity-interpolation} \sloppy
  A finer analysis of algorithm \ip leads to a bit complexity $\gO_\mu(m\log T_1 \mul_\gz(T_1\log(D_1)\log (T_1\log D_1), T_1C_1D_1)$.
\end{rem}

\begin{rem}\label{rem:interpolation}
  Even when \ip returns an incorrect polynomial, it has sparsity at most $2T$, degree less than $D$ and coefficients bounded by $C$.
\end{rem}

\subsection{Multiplication}\label{ssec:multiplication}

Our idea is to compute different candidates to $FG$ with a growing sparsity bound and to verify the result with \verif. Unfortunately, a
direct call to \ip with the correct sparsity $T=\max(\#F,\#G,\#(FG))$ yields a bit complexity $\gOt(T\log(D)\log(CD))$ if the coefficients
are bounded by $C$ and the degree by $D$. We shall remark that it is not nearly optimal since the input and output size are bounded by
$T\log D+T\log C$.

To circumvent this difficulty, we first compute the reductions $F_p = F\bmod X^p-1$ and $G_p = G\bmod X^p-1$ of the input polynomials, as
well as the reductions $F'_p = F'\bmod X^p-1$ and $G'_p=G'\bmod X^p-1$ of their derivatives, for a random prime $p$ as in Corollary
\ref{cor:findterms}. The polynomials $H_p = FG\bmod X^p-1$ and $H'_p = (FG)'\bmod X^p-1$ can be computed using \ip and \verif. Indeed, we
first compute $F_pG_p$ by interpolation and then reduce it \emph{modulo} $X^p-1$ to get $H_p$. Similarly for $H'_p$ we first interpolate
$F'_pG_p+F_pG'_p$ before its reduction. Finally we can compute the polynomial $FG$ from $H_p$ and $H'_p$ using \ft according to Corollary
\ref{cor:findterms}. Our choice of $p$, which is polynomial in the input size, ensures that each call to \ip remains quasi-linear.

\begin{algorithm}
\caption{\spr}
\begin{algorithmic}[1]
  \Input $F,G\in\gz[X]$. $0<\mu_1,\mu_2<1$ with $\frac{\mu_1}{2}\leq\mu_2$.
  \Output $H\in\gz[X]$ s.t. $H=FG$ with probability at least $1-\mu_1$.
  \State $t\gets\max(\#F,\#G)$, $D\gets\deg(F)+\deg(G)$, $C\gets t\|F\|_\infty\|G\|_\infty$
  \State $\lambda\gets\max(21,\frac{20}{3\mu_1}(\#F\#G)^2\ln D)$, $\mu^*\gets\mu_2-\frac{\mu_1}{2}$
  \State $p\gets\rp(\lambda,\frac{\mu_1}{4})$
  \State $F_p\gets F\bmod X^p-1$, $G_p\gets G\bmod X^p-1$
  \State $F'_p\gets F'\bmod X^p-1$, $G'_p\gets G'\bmod X^p-1$
  \Repeat
      \State $N\gets\max(1,\lfloor \frac{32}{5}(t-1)\log p\rfloor)$
      \State $\mathcal{P}\gets \{\text{the first $2N$ primes in increasing order}\}$
      \State $H_1\gets \ip([(F_p,G_p)],t,2p, C,\frac{\mu^*}{2},\mathcal{P})$
      \State $H_2\gets \ip([(F_p,G'_p),(F'_p,G_p)],t,2p, CD,\frac{\mu^*}{2},\mathcal{P})$
      \State $t\gets2t$
      \Until{}
      \Statex \hspace{0.2cm}$\verif(H_1, F_p,G_p,\frac{\mu_1}{2})$ \textbf{and} \Comment $H_1=F_pG_p$
      \Statex \hspace{0.2cm}$\verifsum(H_2,F_p,G'_p,F'_p,G_p,\frac{\mu_1}{2})$  \Comment $H_2=F'_pG_p+F_pG'_p$
  \State $H_p\gets H_1\bmod X^p-1$, $H'_p\gets H_2\bmod X^p-1$.
  \State \Return $\ft\,(p, H_p,H'_p,D,C)$.
\end{algorithmic}
\end{algorithm}

Lemmas~\ref{lem:probasucces} and~\ref{lem:probaquasilinéaire} respectively provide the correctness and complexity bound of algorithm \spr.
Together, they consequently form a proof of Theorem~\ref{thm:product} by taking $\epsilon = \mu_1+\mu_2$. Note that this approach translates
\emph{mutatis mutandis} to the multiplication of sparse polynomials over $\cfq$ where the characteristic of $\cfq$ is larger than $D$.

\begin{lem}\label{lem:probasucces}
  Let $F$ and $G$ be two sparse polynomials over $\gz$. Then algorithm \spr returns $FG$ with probability at least $1-\mu_1$.
\end{lem}

\begin{proof}
  Since $FG$ has sparsity at most $\#F\#G$, Corollary~\ref{cor:findterms} implies that if $H_p=FG \bmod X^p-1$ and $H'_p=(FG)' \bmod X^p-1$,
  the probability that \ft does not return $FG$ is at most $\frac{\mu_1}{2}$. The other reason for the result to be incorrect is that one
  of these equalities does not hold, which means that one of the two verifications fails. Since this happens with probability at most
  $\frac{\mu_1}{2}$, \spr returns $FG$ with probability at least $1-\mu_1$.
\end{proof}

\begin{lem}\label{lem:probaquasilinéaire} \sloppy
  Let $F$ and $G$ be two sparse polynomials over $\gz$, $T=\max(\#F,\#G,\#(FG))$, $D=\deg(FG)$, $C=\max(\|F\|_\infty,\| G \|_\infty,\| FG
  \|_\infty)$ and $\epsilon=\mu_1+\mu_2$. Then algorithm \spr has bit complexity $\gOt_\epsilon(T(\log D+\log C))$ with probability at least
  $1-\mu_2$. Writing $n=T(\log D+\log C)$, the bit complexity is $\gO_\epsilon(n\log^2n\log^2T(\log T+\log\log n))$.
\end{lem}

\begin{proof}
  In order to obtain the given complexity, we first need to prove that with high probability \ip never computes polynomials with a sparsity
  larger than $4\#(FG)$.

  Let $T_p=\max(\#(F_pG_p),\#(F_pG'_p+F'_pG_p))$. If $t\le 2T_p$ then the polynomials $H_1$ and $H_2$ satisfy $\#H_1, \#H_2 \le 4T_p$ by
  Remark~\ref{rem:interpolation}. Unfortunately, $T_p$ could be as large as $T^2$ and $t$ might reach values larger than $T_p$. We now
  prove that: \emph{(i)} with probability at least $1-\mu^*$ the maximal value of $t$ during the algorithm is less than $2T_p$; \emph{(ii)}
  with probability at least $1-\frac{\mu_1}{2}$, $T_p\le\#(FG)$. Together, this will prove that $\#H_1, \#H_2 \le 4\#(FG)$ with probability
  at least $1-\mu^*-\frac{\mu_1}{2} = 1-\mu_2$.

  \emph{(i)} As soon as $t\ge T_p$, Steps~9 and 10 compute both $F_pG_p$ and $F_pG'_p+F'_pG_p$ with probability at least $1-\mu^*$ by
  Theorem~\ref{thm:interpolation}. Since \verif never fails when the product is correct, the algorithm ends when $T_p\leq t< 2T_p$ with
  probability at least $1-\mu^*$.

  \emph{(ii)} Let us define the polynomials $\hat{F}_p$ and $\hat{G}_p$ obtained from $F_p$ and $G_p$ by replacing each nonzero coefficient
  by $1$. The choice of $p$ in Step~3 ensures that with probability at least $1-\frac{\mu_1}{2}$ there is no collision in $(\hat{F}_p
  \hat{G}_p) \bmod X^p-1$ by applying Proposition~\ref{prop:sanscollision} to the product $\hat{F}_p\hat{G}_p$. In that case, there is also
  no collision in $F_pG_p\bmod X^p-1$ and in $F_pG'_p+F'_pG_p\bmod X^p-1$ since $\supp(F_pG_p)\subset\supp(\hat{F}_p\hat{G}_p)$. Therefore,
  there are as many nonzero coefficients in $F_pG_p$ as in $F_pG_p\bmod X^p-1$, which is equal to $FG \bmod X^p-1$. Thus with probability at
  least $1-\frac{\mu_1}{2}$ we have $\#(F_pG_p) = \#(FG)\le T$ and similarly $\#(F'_pG_p+F_pG'_p) = \#((FG)')\le T$.

  In the rest of the proof, we assume that the loop stops with $t\le 2T_p$ and that $T_p\le T$. In particular, the number of iterations of
  the loop is $\gO(\log T)$. Since $2p=\gO(\frac{1}{\epsilon}T^4\log D)$, Steps~9 and 10 have a bit complexity
  $\gOteps(T\log(p)\log(pCD))=\gOteps(T\log CD)$ by Theorem~\ref{thm:interpolation}. Using Remark \ref{rem:interpolation}, \verif and
  \verifsum have polynomials of height at most $tCD$ as inputs. By Corollary~\ref{cor:verifsumprodz}, Step~12 has bit complexity
  $\gOeps(T\log(T\log p)\imul(\log CD))=\gOteps(T\log CD)$. The list $\mathcal P$ can be computed incrementally, adding new primes when
  necessary. At the end of the loop, $\mathcal{P}$ contains $\gO(T\log 2p)$ primes, which means that it is computed in
  $\gOeps(T\log(p)\log^2(T\log p)\log\log(T\log p))$ bit operations~\cite[Chapter 18]{MCAlgebra}, that is $\gOteps(T\log\log D)$ since $\log
  p=\gO(\log(T\log D))$.

  \sloppy
  The total cost for the $\gO(\log T)$ iterations of the loop is still $\gOteps(T\log(CD))$. Step~14 runs in time $\gO_\epsilon(T\imul(\log
  CD))$ by Lemma~\ref{lem:findterms} as the coefficients of $H'_p$ are bounded by $2TC^2D$ with $T\le D$ and $\#H_p,\#H'_p \le \#H$. Since
  other steps have negligible costs this yields a complexity of $\gOteps(T(\log C+\log D))$ with probability at least $1-\mu_2$.

  Using Remark~\ref{rem:complexity-interpolation}, we can provide a more precise complexity for Steps~9 and 10 which is $\gOeps(\log T
  \mul_\gz(T\log(p)\log (T\log p), pDTC))$ bit operations. It is easy to observe that the $\log T$ repetitions of these steps provide the
  dominant term in the complexity. A careful simplification yields a bit complexity $\gOeps(n\log^2 n \log^2 T(\log T+\log\log n)$ for \spr
  where $n=T(\log D+\log C)$ bounds both input and output sizes.
\end{proof}

\subsection{Multivariate case}
Using classical Kronecker substitution~\cite[Chapter 8]{MCAlgebra} one can extend straightforwardly \spr to multivariate polynomials. Let
$F,G \in\gz[X_1,\ldots,X_n]$ with $\|F\|_\infty, \|G\|_\infty \le C$ and $\deg_{X_i}(F)+\deg_{X_i}(G) < d$. Writing $F_u(X) = F(X, X^d,
\dotsc, X^{d^{n-1}})$ and $G_u(X) = G(X, X^d, \dotsc, X^{d^{n-1}})$, one can easily retrieve $FG$ from the univariate product $F_uG_u$. It
is easy to remark that the Kronecker substitution preserve the sparsity and the height, and it increases the degree to $\deg F_u, \deg G_u
<d^n$. If $F$ and $G$ are sparse polynomials with at most $T$ nonzero terms, their sizes are at most $T(n\log d+\log C)$ which is exactly
the sizes of $F_u$ and $G_u$. Since the Kronecker and inverse Kronecker substitutions cost $\gOt(Tn\log d)$ bit operations, one can compute
$F_uG_u$ using \spr within the following bit complexity.

\begin{cor}\label{cor:multivariatesparsemuloverintegers}
  There exists an algorithm that takes as inputs $F,G\in\gz[X_1,\dotsc,X_n]$ and $0 < \epsilon < 1$, and computes $FG$ with probability at
  least $1-\epsilon$, using $\gOt_\epsilon(T(n\log d+\log C))$ bit operations where $T=\max(\#F,\#G,\#(FG))$, $d=\max_i(\deg_{X_i} FG)$ and
  $C=\max(\|F\|_\infty, \|G\|_\infty)$.
\end{cor}

Over a finite field $\cfqs$ for some prime $q$, the previous technique requires that $q>d^n$ since \spr requires $q$ to be larger than the
degree. The \emph{randomized Kronecker substitution} method introduced by Arnold and Roche~\cite{ArRo14} allows to apply \spr to fields of
smaller characteristic. The idea is to define univariate polynomials $F_s(X) = F(X^{s_1},\dotsc,X^{s_n})$ and $G_s(X) = G(X^{s_1},\dotsc,
X^{s_n})$ for some random vector $\vec s=(s_1,\dots,s_n)$ such that these polynomials have much smaller degrees than those obtained with
classical Kronecker substitution. As a result, we obtain an algorithm that works for much smaller $q$ of order $\gOt(nd\#F\#G)$.

Our approach is to first use some randomized Kronecker substitutions to estimate the sparsity of $FG$ by computing the sparsity of $H_s =
F_sG_s$ for several distinct random vectors $\vec s$. With high probability, the maximal sparsity is close to the one of $FG$. Then, we use
this information to provide a bound to some (multivariate) sparse interpolation algorithm. Note that our approach is inspired
from~\cite{HuangGao2019} that slightly improves randomized Kronecker substitution.

\begin{lem}\label{lem:sparsityestimate}
  Let $H\in\cfqs[X_1,\dotsc,X_n]$ of sparsity $T$, and $\vec s$ be a vector chosen uniformly at random in $S^n$ where $S\subset \N$ is
  finite. The expected sparsity of $H_s(X)=H(X^{s_1}, \dotsc,X^{s_n})$ is at least $T(1-\frac{T-1}{\#S})$.
\end{lem}

\begin{proof}
  If we fix two distinct exponent vectors $\vec e_u$ and $\vec e_v$ of $H$, they \emph{collide} in $H_s$ if and only if $\vec e_u\cdot\vec s
  = \vec e_v\cdot\vec s$. Since $\vec e_u\neq \vec e_v$, they differ at least on one component, say $e_{u,j_0}\neq e_{v,j_0}$. The equality
  $\vec e_u\cdot\vec s = \vec e_v\cdot\vec s$ is then equivalent to
  \[ s_{j_0} = \sum_{j\neq j_0} \frac{e_{v,j}-e_{u,j}}{e_{u,j_0}-e_{v,j_0}} s_j. \]
  Writing $Y$ for the right-hand side of this equation we have
  \[\proba[\vec e_u\cdot\vec s = \vec e_v\cdot\vec s] = \proba[s_{j_0} = Y] = \sum_y \proba[s_{j_0} = Y| Y = y]\proba[Y = y]\]
  where the (finite) sum ranges over all possible values $y$ of $Y$. Since $s_{j_0}$ is chosen uniformly at random in $S$, $\proba[s_{j_0}
  = Y|Y = y] = \proba[s_{j_0} = y]\le 1/\#S$ and the probability that $\vec e_u$ and $\vec e_v$ collide is at most $1/\#S$. This implies
  that the expected number of vectors that collide is at most $T(T-1)/\#S$.
\end{proof}

\begin{cor}\label{cor:sparsityestimate}
  Let $H$ be as in Lemma~\ref{lem:sparsityestimate} and $\vec v_1$, \dots, $\vec v_\ell\in S^n$ be some vectors chosen uniformly and
  independently at random. Then $\proba[\max_i \#H_{v_i}\le T(1-2\frac{T-1}{\#S})]\le 1/2^\ell$.
\end{cor}

\begin{proof}
  For each $\vec v_i$, the expected number of terms that collide in $H_{v_i}(X)$ is at most $T(T-1)/\#S$ by
  Lemma~\ref{lem:sparsityestimate}. Using Markov's inequality, we have $\proba[\#H_{v_i}\le T-2T(T-1)/\#S] \le 1/2$. Since the vectors $\vec
  v_i$ are independent, the result follows.
\end{proof}

\begin{algorithm}
\algrenewcommand\algorithmicloop{\textbf{repeat}}
\caption{\sparsityEstimate}
\begin{algorithmic}[1]
  \Input $F, G\in \cfqs[X_1,\dotsc,X_n]$, $0<\epsilon<1$, $\lambda>1$.
  \Output An integer $t$ such that $t\le \lambda\#(FG)$.
  \State $N\gets \lceil2\frac{\#F\#G-1}{1-1/\lambda}\rceil$, $\ell\gets \lceil \log\frac{2}{\epsilon}\rceil$.
  \State $t'\gets 0$, $\mu\gets\frac{\epsilon}{4\ell}$.
  \Loop~$\ell$ times
    \State $\vec s\gets$ random element of $\{0,\dots,N-1\}^n$.
    \State $F_s\gets F(X^{s_1},\ldots,X^{s_n})$, $G_s\gets G(X^{s_1},\ldots,X^{s_n})$
    \State $H_s\gets \spr(F_s, G_s, \mu, \mu)$
    \State $t'\gets \max(t',\#H_s)$
  \EndLoop
  \State \Return $\lambda t'$.
\end{algorithmic}
\end{algorithm}

\begin{lem}
  Algorithm \sparsityEstimate is correct when $q\ge \frac{4D\#F\#G}{1-1/\lambda}$ where $D = \max(\deg F,\deg G)$. With probability at
  least $1-\epsilon$, it returns an integer $t\ge\#(FG)$ using $\gOt_\epsilon(T(n\log d+s\log q))$ bit operations where
  $T=\max(\#(FG),\#F,\#G)$ and $d=\max_i( \deg_{X_i} FG)$.
\end{lem}

\begin{proof}
  Since each polynomial $H_s$ has sparsity at most $\#(FG)$, \sparsityEstimate returns an integer bounded by $\lambda\#(FG)$. \spr can be
  used in step~5 since $\deg H_s = \deg F_s+\deg G_s \le 2ND \le q$ by the definition of $N$. Assuming that \spr returns no incorrect
  answer during the loop, Corollary~\ref{cor:sparsityestimate} applied to the product $FG$ implies that $t'\ge \#(FG)(1 - 2(\#(FG)-1)/N)$
  with probability $\ge1-\epsilon/2$ at the end of the loop. By definition of $N$ and since $\#F\#G \ge \#(FG) $, $t'\ge \#(FG)/\lambda$.
  Taking into account the probability of failure of \spr, the probability that $\lambda t' \ge\#(FG)$ is at least $1-\frac{3\epsilon}{4}$.

  The computation of $F_s$ and $G_s$ requires $\gO(Tn\imul(\log\max(d,N))+Ts\log q)$ bit operations in Step~5. Since
  $\max(\#F_s,\#G_s,\#H_s) \le T$ and $\deg H_s=\gO(ndT^2)$ in Step~6, the bit complexity of each call to \spr is $\gOt_\mu(T(\log(nd)+s\log
  q))$ with probability at least $1-\mu$ using Lemma~\ref{lem:probaquasilinéaire}. Therefore, \sparsityEstimate requires
  $\gOt_\epsilon(T(n\log d+s\log q))$ bit operations with probability at least $1-\epsilon/4$. Together with the probability of failure this
  concludes the proof.
\end{proof}

\begin{thm}
  There exists an algorithm that takes as inputs two sparse polynomials $F$ and $G$ in $\cfqs[X_1,\dotsc,X_n]$ and $0<\epsilon <1$ that
  returns the product $FG$ in $\gOt_\epsilon(nT(\log d+s\log q))$ bit operations with probability at least $1-\epsilon$, where $T =
  \max(\#F, \#G, \#(FG))$, $d=\max_i(\deg_{X_i} FG)$, $D=\deg FG$ and assuming that $q = \Omega(D\#F\#G + DT\log(D)\log(T\log D))$.
\end{thm}

\begin{proof}
  The algorithm computes an estimate $t$ on the sparsity of $FG$ using $\sparsityEstimate(F,G,\frac{\epsilon}{2},\lambda)$ for some constant
  $\lambda$. The second step interpolates $FG$ using Huang and Gao's algorithm \cite[Algorithm 5 (\textsc{MulPolySI})]{HuangGao2019} which
  is parameterized by a univariate sparse interpolation algorithm. Originally, its inputs are a polynomial given as a blackbox and bounds
  on its degree and sparsity. In our case, the blackbox is replaced by $F$ and $G$, the sparsity bound is $t$ and the univariate
  interpolation algorithm is \spr.

  The algorithm \textsc{MulPolySI} requires $\gO_\epsilon(n\log t+\log^2 t)$ interpolation of univariate polynomials with degree $\gOt(tD)$
  and sparsity at most $t$. Each interpolation with $\spr$ is done with $\mu_1,\mu_2$ such that $\mu_1+\mu_2={\epsilon}/{4(n+1)\log t}$, so
  that \textsc{MulPolySI} returns the correct answer in $\gOt_\epsilon(nT(\log d+s\log q))$ bit operations with probability at least
  $1-\frac{\epsilon}{2}$~\cite[Theorem~6]{HuangGao2019}. Altogether, our two-step algorithm returns the correct answer using
  $\gOt_\epsilon(nT(\log d+s\log q))$ bit operations with probability at least $1-\epsilon$. The value of $q$ is such that it bounds the
  degrees of the univariate polynomials returned by \spr during the algorithm.
\end{proof}

\subsection{Small characteristic}

We now consider the case of sparse polynomial multiplication over a field $\cfqs$ with characteristic smaller than the degree of the product
$FG$ (or, in the multivariate case, smaller than the degree of the product after randomized Kronecker substitution). We can no more use
Huang's interpolation algorithm since it uses the derivative to encode the exponents into the coefficients and thus it only keeps the value
of the exponents \emph{modulo} $q$. Our idea to circumvent this problem is similar to the one in~\cite{roche2015} that is to rather
consider the polynomials over $\gz$ before calling our algorithm \spr.

The following proposition is only given for the multivariate case as it encompasses univariate's one. It matches exactly with the complexity
result given by Arnold and Roche~\cite{roche2015}.

\begin{prop} \sloppy
  There exists an algorithm that takes as inputs two sparse polynomials $F$ and $G$ in $\cfqs[X_1,\ldots,X_n]$ and $0<\epsilon<1$ that
  returns the product $FG$ in $\gOt_\epsilon (S(n\log d+s\log q))$ bit operations with probability at least $1-\epsilon$, where $S$ is the
  structural sparsity of $FG$ and $d=\max_i(\deg_{X_i} FG)$.
\end{prop}

\begin{proof}
  If $s=1$, the coefficients of $F$ and $G$ map easily to the integers in $\{0,\ldots,q-1\}$. Therefore, the product $FG$ can be obtained by
  using an integer sparse polynomial multiplication, as the one in Corollary~\ref{cor:multivariatesparsemuloverintegers}, followed by some
  reductions \emph{modulo} $q$. Unfortunately, mapping the multiplication over the integers implies that the cancellations that could have
  occurred in $\cfq$ do not hold anymore. Consequently, the support of the product in $\gz$ before modular reduction is exactly the
  structural support of $FG$.

  If $s>1$, the coefficients of $F$ and $G$ are polynomials over $\cfq$ of degree $s-1$. As previously, mapping $\cfq$ to integers, $F$ and
  $G$ can be seen as $F_Y,G_Y\in\gz[Y][X_1,\dots,X_n]$ where the coefficients are polynomials in $\gz[Y]$ of degree at most $s-1$ and height
  at most $q-1$.

  If $T = \max(\#F,\#G)$, the coefficients of $F_YG_Y$ are polynomials of degree at most $2s-2$ and height at most $Tsq^2$. Therefore, the
  product $FG\in\cfqs$ can be computed by: \emph{(i)} computing $F_B,G_B \in \gz[X_1,\dots,X_n]$ by evaluating the coefficients of $F_Y$ and
  $G_Y$ at $B=Tsq^2$ (Kronecker substitution); \emph{(ii)} computing the product $H_B=F_BG_B$; \emph{(iii)} writing the coefficients of
  $H_B$ in base $B$ to obtain $H_Y=F_YG_Y$ (Kronecker segmentation); \emph{(iv)} and finally mapping back the coefficients of $H_Y$ from
  $\gz[Y]$ to $\cfqs$.

  Similarly as the case $s=1$, $H_B$ and then $H_Y$ have at most $S$ nonzero coefficients. The Kronecker substitutions in \emph{(i)}
  require $\gOt(Ts\log q)$ bit operations, while the Kronecker segmentations in \emph{(iii)} need $\gOt(Ss\log q)$ bit operations. In
  \emph{(iv)} we first compute $Ss$ reductions \emph{modulo} $q$ on integers smaller than B, and then $S$ polynomial divisions in $\cfq[Y]$
  with polynomial of degree $\gO(s)$. Thus, it can be done in $\gOt(Ss\log q)$ bit operations. Finally the computation in \emph{(ii)} is
  dominant and it requires $\gOt_\epsilon (S(n\log d+s\log q))$ bit operations with probability at least $1-\epsilon$ using
  Corollary~\ref{cor:multivariatesparsemuloverintegers}.
\end{proof}

\newcommand{\Gathen}{\relax}\newcommand{\Hoeven}{\relax}

\end{document}